\documentclass[11pt]{article}

\usepackage[utf8]{inputenc}
\usepackage[T1]{fontenc}
\usepackage[a4paper,margin=2.5cm]{geometry}
\usepackage{amsmath,amssymb,amsthm}
\usepackage{graphicx}
\usepackage{booktabs}
\usepackage[table]{xcolor}
\usepackage{array}
\usepackage[labelfont=bf,labelsep=period,font=small]{caption}
\usepackage{authblk}
\usepackage{enumitem}
\usepackage[colorlinks=true,linkcolor=blue!70!black,citecolor=blue!70!black,urlcolor=blue!70!black]{hyperref}

\newcommand{\norm}[1]{\lVert #1\rVert}
\newcommand{\ip}[2]{\langle #1, #2\rangle}
\newcommand{\Com}{\mathsf{Com}}
\newtheorem{definition}{Definition}
\newtheorem{lemma}{Lemma}
\newtheorem{proposition}{Proposition}

\usepackage[numbers,sort&compress]{natbib}
\usepackage[ruled,linesnumbered]{algorithm2e}

\definecolor{recrow}{gray}{0.90}

\title{PVBRF-ID: Private, Verifiable and Byzantine-Robust Federated Intrusion Detection for Oilfield Enterprise Networks}

\author[1]{Mustafa S. Aljumaily}
\author[2]{Ali Ataeemh Allami}
\author[2]{Mustafa N. Mnati}
\author[2]{Nawar S. Alseelawi}
\author[3]{Shahab Abdulla}
\affil[1]{Daw Alfada Company, Amarah, Maysan, Iraq \\ \href{mailto:mustafa.s@daw-alfada.com}{mustafa.s@daw-alfada.com}}
\affil[2]{University of Misan, Amarah, Maysan, Iraq \\ \href{mailto:ali_ataeemh@uomisan.edu.iq}{ali\_ataeemh@uomisan.edu.iq} \quad \href{mailto:Mustafa.n@uomisan.edu.iq}{Mustafa.n@uomisan.edu.iq} \quad \href{mailto:nawar.alseelawi@uomisan.edu.iq}{nawar.alseelawi@uomisan.edu.iq}}
\affil[3]{University of Southern Queensland, Toowoomba, QLD, Australia \\ \href{mailto:Shahab.abdulla@unisq.edu.au}{Shahab.abdulla@unisq.edu.au}}
\date{}

\begin{document}
\maketitle

\begin{abstract}
\noindent Federated learning allows dispersed oilfield sites to train a shared intrusion detector without pooling network traffic, yet its aggregation step forces a choice between privacy and robustness. Byzantine-robust rules must read each update in the clear, whereas secure aggregation hides updates from the checks that detect poisoning. We present PVBRF-ID, a two-server protocol in which sites secret-share their updates and the servers jointly test each update's norm and its cosine similarity to a reference computed on a small clean root set, revealing one accept bit per site. Commitments, Merkle roots and a homomorphic audit make the inclusion of accepted updates publicly verifiable. On a 15-class Edge-IIoTset task with ten sites holding differently distributed data and ten random seeds, PVBRF-ID matched FedAvg on clean data (macro-F1 0.708 vs.\ 0.706) and withstood sign-flip attacks (0.696 vs.\ 0.091). Against an adaptive attacker controlling 30\% of sites, PVBRF-ID with secure clipping reached 0.641, whereas FedAvg fell to 0.046; plaintext multi-Krum remained higher by 0.052. An update-inference attacker recovered each site's class mix from plaintext updates (Spearman correlation 0.81) but not from one server's view. Backdoors passed the filter, and noise-based mitigation gave only heuristic protection. The implementation uses Python 3.13, NumPy 2.3, 30 rounds and 10 seeds.
\end{abstract}


\section*{Introduction}

Oilfield operators run geographically dispersed sites, such as wellheads, gathering stations and pipeline segments, whose operational-technology (OT) and Internet-of-Things (IoT) traffic is heterogeneous. Integrating information technology (IT) with OT in these assets has widened their exposure to cyber attacks\cite{progoulakis2021offshore}. Bandwidth limits, data-sovereignty rules and operational sensitivity keep this traffic at the site, so an intrusion detection system (IDS) trained on pooled traffic is rarely an option\cite{khraisat2024flids}. Federated learning (FL)\cite{mcmahan2017fedavg} lets such sites train one detector while each keeps its traffic local. Public industrial IoT (IIoT) datasets such as Edge-IIoTset\cite{ferrag2022edge} were built with both centralized and federated evaluation in mind.

The aggregation step of FL pulls two requirements apart. Compromised sites can poison the shared model\cite{bagdasaryan2020backdoor,sun2019backdoor,fang2020local}, and the standard defences inspect each update in plaintext. Updates also leak the class mix and other properties of a site's data\cite{melis2019leakage,zhu2019dlg}. Secure aggregation\cite{bonawitz2017secagg} hides individual updates from the server, and with them the evidence a robust rule needs. This privacy--robustness tension has been documented across several federated learning settings\cite{lyu2022robustness}; here we build and evaluate an aggregator that resolves part of it.

\textit{Secure and verifiable aggregation.} Pairwise-masking secure aggregation\cite{bonawitz2017secagg} hides each update from the server. Prio\cite{corrigan2017prio} secret-shares client submissions across non-colluding servers and validates them with proofs. RoFL\cite{lycklama2023rofl} combines commitments with zero-knowledge proofs that bound the $L_2$ norm of an update, and ELSA\cite{rathee2023elsa} aggregates with two servers while tolerating malicious participants. Prio, RoFL and ELSA bound the size or check the well-formedness of an update, but none tests its direction against a trusted reference. FLOD\cite{dong2021flod} adds such a test: two non-colluding semi-honest servers use secure two-party computation to compare binarized client updates with an update from a small root dataset, and they reveal only the weighted aggregate. ShieldFL\cite{ma2022shieldfl} measures cosine similarity between gradients encrypted under two-trapdoor homomorphic encryption. Neither FLOD nor ShieldFL gives a client a way to verify that its accepted update entered the published aggregate, and FLOD's evaluation uses Gaussian and label-flipping attacks on image benchmarks rather than an attacker adapted to the test itself. Privacy-preserving aggregation in networked and transportation settings has been studied alongside these protocols\cite{allami2018watermark,nicewarner2024aggregation}.

\textit{Byzantine-robust aggregation.} Krum and multi-Krum\cite{blanchard2017krum} keep the updates closest to their neighbours, and coordinate-wise trimmed mean\cite{yin2018trimmed} discards extreme coordinates. FLTrust\cite{cao2021fltrust} weights each update by its cosine similarity to an update computed on a small clean root dataset. Adaptive attacks crafted against a known defence\cite{fang2020local} show that each rule must face an attacker who knows it. All three rules read every update in plaintext, so none runs inside secure aggregation without change.

\textit{Backdoors, leakage and differential privacy.} Boosted model replacement injects backdoors that survive averaging\cite{bagdasaryan2020backdoor}; norm clipping and added noise reduce their effect but do not remove it\cite{sun2019backdoor}. Leakage from gradients and updates\cite{melis2019leakage,zhu2019dlg} motivates differentially private (DP) training, which clips updates and adds Gaussian noise\cite{abadi2016dpsgd,mcmahan2018dpfedavg} and is accounted with R\'enyi DP\cite{mironov2017rdp}. Federated IDS for IoT and IIoT are surveyed in~\cite{nguyen2021fliot}. DeepFed\cite{li2021deepfed} encrypts model parameters with the Paillier cryptosystem but does not test robustness to poisoning. Rashid \textit{et al.}\cite{rashid2023fliiot} evaluate FedAvg on Edge-IIoTset under a trusted aggregator and leave poisoning defences to future work. Neither measures what the aggregator learns about each site or whether aggregation excludes sites with atypical traffic. Oblivious policy enforcement and privacy-preserving task coordination in industrial and mobile network settings\cite{allami2025firewall,nicewarner2026crowdsensing} motivate extending these protections to the aggregation layer.

We ask whether one aggregator can be private and robust at the same time, at a cost an operator can pay. Our contributions are:
\begin{enumerate}
  \item \textbf{Design.} We introduce PVBRF-ID, a two-server, secret-shared aggregator. The servers test each update's norm and a tolerant cosine against a root-set reference and reveal one accept bit per site. Published commitments, Merkle roots and an audit make the inclusion of accepted updates verifiable.
  \item \textbf{A design finding.} A reference taken from the previous aggregate is public and attacker-influenced. An adaptive attacker captures it and drives macro-F1 to 0.012; a clean root set avoids this failure.
  \item \textbf{A worst-case evaluation.} We benchmark PVBRF-ID on a 15-class Edge-IIoTset task (ten sites, ten seeds, paired statistics) against FedAvg, trimmed mean, multi-Krum and norm-only aggregation, under five non-adaptive attacks and a boundary-seeking adaptive attacker that exploits the cosine tolerance.
  \item \textbf{Measured privacy and fairness.} An inference attacker recovers class mix from plaintext updates. Multi-Krum systematically excludes atypical sites, whereas PVBRF-ID with a suitable norm bound does not.
  \item \textbf{Negative results.} Backdoors pass the filter, noise-based mitigation is heuristic, and plaintext multi-Krum is more robust at 30\% adversaries.
\end{enumerate}

Experiments use public IIoT traffic with simulated sites, and no operator data are used. Four cryptographic components are modelled as ideal functionalities. The Methods section defines the threat model (Definition~\ref{def:threat}), the protocol and the experimental setup. The Results section reports accuracy, robustness, privacy, fairness and cost. The Discussion states when PVBRF-ID is the right choice and what bounds its scope.

\section*{Methods}

\subsection*{System, threat model and goals}

There are $n$ sites, each holding labelled local traffic. Two servers $S_1$ and $S_2$ operate in separate administrative domains, for example a corporate security operations centre and a regional or contractor domain, and share a public bulletin board. The servers hold a small clean root set, for example captured on a testbed.

\begin{definition}[Threat model]\label{def:threat}
The protocol operates against three adversary classes.

\noindent\textbf{(i) Byzantine sites.} Up to $f < n/2$ of the $n$ sites are Byzantine. A Byzantine site may submit any vector as its update, coordinate with other Byzantine sites, and adapt its strategy to the full public state, including the bound $B_t$, the threshold $\tau$, the published reference direction $\hat{r}_t$, and the accept bitmaps of previous rounds. Byzantine sites cannot forge another site's Pedersen commitment, break the secret-sharing randomness, or observe a honest site's local traffic.

\noindent\textbf{(ii) Honest-but-curious servers.} Each server $S_j$ follows the protocol exactly but records everything it sees and tries to infer individual site updates from its view: its own shares $\{s_i^j\}$, the accept bits $\{b_i\}$, the other server's published aggregate share, and all public transcripts. The two servers do not collude and do not share private state. A server cannot deviate from the secure-test protocol or alter an accept bit without detection by the audit.

\noindent\textbf{(iii) External verifier / transcript attacker.} A passive external adversary observes the public bulletin board in every round: the published aggregate $\Delta_t$, the accept bitmap $\{b_i\}$, the Merkle roots, and the global model sequence $\{w_t\}$. This adversary tries to infer site-level class distributions from the published transcript; it cannot see either server's private shares.

\noindent\textbf{Out of scope.} Collusion between $S_1$ and $S_2$; compromise of the root set; denial-of-service against the bulletin board; Sybil enrolment of new sites; and membership inference or model inversion from the released global model weights.
\end{definition}

Under Definition~\ref{def:threat}, PVBRF-ID targets five goals. (G1) Clean accuracy is close to FedAvg. (G2) The model is robust to the attackers of the Attacks subsection. (G3) One server learns nothing about an individual update beyond one accept bit per site per round and the aggregate. (G4) Inclusion of accepted contributions is verifiable. (G5) Honest but atypical sites are not systematically excluded.

\subsection*{Protocol overview}

In each round $t$ (Fig.~\ref{fig:protocol}), the servers compute a root update $r_t$ from the root set and the current global model. Each site $i$ computes its update $u_i$ and sends one additive share to each server, together with commitments. The servers jointly test every update and learn only an accept bit $b_i$. Accepted updates, optionally clipped, are summed, the aggregate is revealed, and a transcript is published.

\begin{figure}[t]
  \centering
  \includegraphics[width=0.62\linewidth]{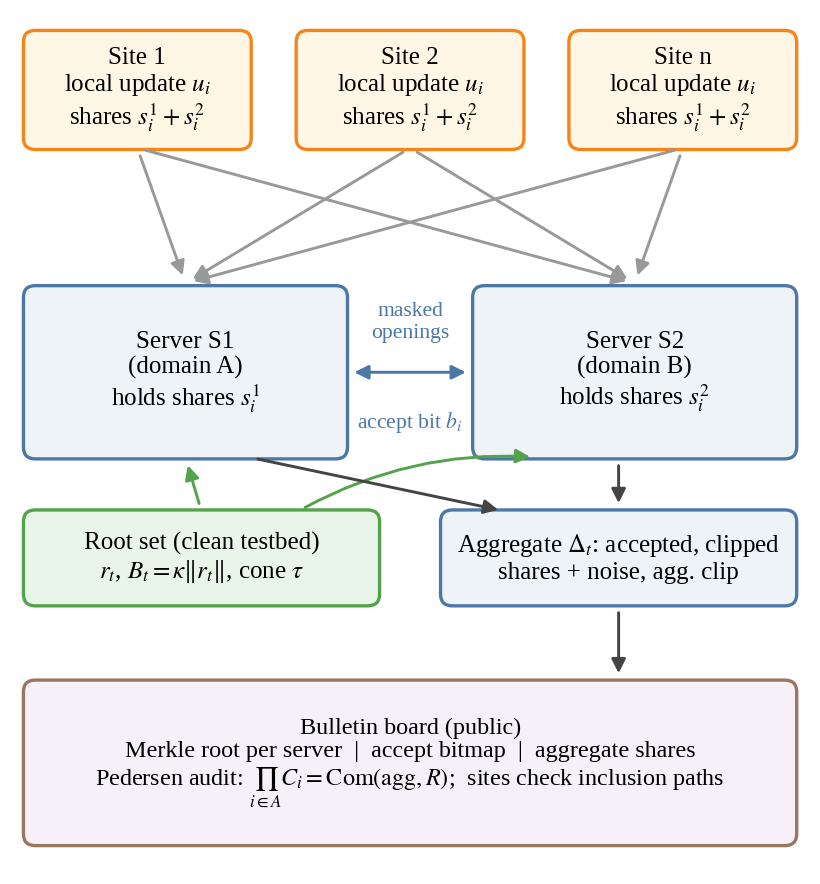}
  \caption{\textbf{PVBRF-ID.} Sites secret-share their updates to two servers in separate administrative domains. The servers test each update against a root-set reference and reveal one accept bit per site. Per-server Merkle roots, the accept bitmap, the aggregate shares and a Pedersen audit are published on a bulletin board, and each site checks the inclusion path of its commitment.}
  \label{fig:protocol}
\end{figure}

\begin{algorithm}[t]
\DontPrintSemicolon
\caption{PVBRF-ID aggregation round $t$}
\label{alg:pvbrf}
\KwIn{Global model $w_t$; root set $\mathcal{R}$; public parameters $\kappa$, $\tau$, learning rate $\eta$}
\KwOut{Updated model $w_{t+1}$; public transcript $(B_t,\hat{r}_t,\{b_i\},A_1,A_2,R_1,R_2,\text{Merkle roots})$}
\BlankLine
\tcc{--- Server setup ---}
Servers jointly train on $\mathcal{R}$ with $w_t$ to obtain root update $r_t$\;
Publish $B_t \leftarrow \kappa\norm{r_t}$ and $\hat{r}_t \leftarrow r_t/\norm{r_t}$\;
Dealer supplies $n$ Beaver triples $(\langle a_i\rangle, \langle \norm{a_i}^2\rangle)$ and $n$ scalar triples for cosine tests\;
\BlankLine
\tcc{--- Site submission (all sites $i = 1,\dots,n$ in parallel) ---}
\ForEach{site $i$}{
  Compute local update $u_i \leftarrow \mathrm{SGD}(w_t, \mathcal{D}_i)$\;
  Sample $s_i^1 \overset{\$}{\leftarrow} \mathbb{F}_Q^d$;\quad set $s_i^2 \leftarrow (u_i - s_i^1) \bmod Q$\;
  Generate Pedersen commitments $C_i^j \leftarrow \Com(s_i^j;\,\rho_i^j)$ for $j \in \{1,2\}$\;
  Send $(s_i^j,\, C_i^j,\, \rho_i^j)$ to server $S_j$\;
}
\BlankLine
\tcc{--- Per-site checks (servers, all $i$ in parallel) ---}
\ForEach{site $i$}{
  $S_j$ verifies $C_i^j$ opens to $s_i^j$; set $b_i \leftarrow 0$ and skip on failure\;
  Range proofs (idealized) confirm $|u_{i,k}| \le 2^{20}$ for all coordinates $k$; reject on failure\;
  Servers run Beaver protocol on shares to obtain $\langle\norm{u_i}^2\rangle$ and $\langle \delta_i^2\rangle$ where $\delta_i = \langle u_i, \hat{r}_t\rangle$\;
  Ideal comparison reveals accept bit $b_i$ satisfying \eqref{eq:accept}--\eqref{eq:tolerant}\;
}
\BlankLine
\tcc{--- Aggregation and publish ---}
\ForEach{server $j \in \{1,2\}$}{
  $A_j \leftarrow \sum_{i:\,b_i=1} s_i^j$ \quad (optional: apply secure clipping and Gaussian noise inside sum)\;
  $R_j \leftarrow \sum_{i:\,b_i=1} \rho_i^j$\;
  Publish Merkle root over $\{C_i^j\}$, bitmap $\{b_i\}$, aggregate share $A_j$, and $R_j$\;
}
\BlankLine
\tcc{--- Reveal, audit, and model update ---}
$\Delta_t \leftarrow (A_1 + A_2)\;/\;|\mathcal{A}_t|$\;
$w_{t+1} \leftarrow w_t + \eta\,\Delta_t$\;
Anyone verifies equation~\eqref{eq:audit} and Merkle paths for each accepted site $i$\;
\end{algorithm}

\subsection*{Secret sharing}

Each coordinate is encoded as a fixed-point integer with 16 fractional bits in the prime field $\mathbb{F}_Q$, $Q=2^{127}-1$, and is bounded by $|x|\le2^{20}$ through a per-coordinate range proof (idealized). A site splits $u_i = s_i^1 + s_i^2$ with $s_i^1$ uniform, so each share alone is uniformly distributed and reveals nothing about $u_i$.

\subsection*{Secure acceptance tests}

Let $\hat r_t = r_t/\norm{r_t}$ be the unit reference. The public norm bound is
\begin{equation}\label{eq:bound}
  B_t = \kappa\,\norm{r_t},
\end{equation}
where $\kappa$ is a public multiplier. An update is accepted if and only if
\begin{equation}\label{eq:accept}
  \underbrace{\norm{u_i}\le B_t}_{\text{bounded influence}} \;\wedge\; \underbrace{\cos(u_i, r_t)\ge\tau}_{\text{alignment with the root reference}},
\end{equation}
where $\tau$ is the public cosine threshold. With $\delta_i=\ip{u_i}{\hat r_t}$, the cosine condition in \eqref{eq:accept} becomes $\delta_i\ge0 \wedge \delta_i^2\ge\tau^2\norm{u_i}^2$ for $\tau\ge0$. For the tolerant case $\tau<0$, which admits mildly misaligned honest updates from sites with skewed class mixes, it becomes
\begin{equation}\label{eq:tolerant}
  \delta_i\ge0 \;\vee\; \delta_i^2\le\tau^2\norm{u_i}^2 .
\end{equation}
The inner product $\delta_i$ is linear in the shares and is computed locally by each server. The squares $\norm{u_i}^2$ and $\delta_i^2$ need one Beaver multiplication each\cite{beaver1991}. For the norm, the servers open $e=u_i-a$ for a dealer-supplied random vector $a$, which makes $e$ uniform, and compute
\begin{equation}\label{eq:beaver}
  \norm{u_i}^2 = \underbrace{\ip{e}{e}}_{\text{public, opened}} + \underbrace{2\ip{e}{a}}_{\text{linear in shares of } a} + \underbrace{\ip{a}{a}}_{\text{dealer-supplied share}} .
\end{equation}
Each server holds a share of every term in \eqref{eq:beaver}, and the sign tests in \eqref{eq:accept} and \eqref{eq:tolerant} are performed by an ideal comparison functionality that reveals only $b_i$. Because $\norm{r_t}$ shrinks as training converges, the bound in \eqref{eq:bound} adapts to the training stage. Lemma~\ref{lem:wrap} guarantees that the fixed-point arithmetic behind these tests never wraps around the field modulus. We calibrated $\kappa$ on the first six clean FedAvg rounds of seed 0, which is also one of the ten evaluation seeds. Over these rounds, the ratio $\norm{u_i}/\norm{r_t}$ for honest sites had median 1.22, 99th percentile 1.93 and maximum 2.07, and we fixed $\kappa=3.5$.

\begin{lemma}[No wrap-around]\label{lem:wrap}
Let every coordinate of $u_i$ satisfy $|u_{i,k}|\le2^{20}$ and be encoded with 16 fractional bits in $\mathbb{F}_Q$, $Q=2^{127}-1$. Let $|\tau|\le1$, $B_t<2^{47}$ and let the model dimension satisfy $d<2^{19}$. Then every quantity whose sign is tested in \eqref{eq:accept} and \eqref{eq:tolerant} has absolute value below $(Q-1)/2$. The signed reconstruction therefore equals the integer value, and the secure tests return the same bit as the plaintext rule applied to the encoded update.
\end{lemma}
\begin{proof}
Encoded coordinates satisfy $|\tilde u_k|\le2^{36}$, so $\norm{\tilde u}^2\le d\,2^{72}$. The encoded unit reference has $|\tilde r_k|\le2^{16}|\hat r_{t,k}|+\tfrac12$, so $\norm{\tilde r}\le2^{16}+\sqrt d/2\le2^{17}$. By the Cauchy--Schwarz inequality, $|\tilde \delta_i|\le\sqrt d\,2^{53}$ and $\tilde \delta_i^{\,2}\le d\,2^{106}$. The encoded threshold satisfies $\tilde\tau^2\le2^{32}$, so $\tilde\tau^2\norm{\tilde u}^2\le d\,2^{104}$. Every tested difference is therefore below $d\,2^{107}<2^{126}$ for $d<2^{19}$. The norm test compares $\norm{\tilde u}^2$ with $\tilde B_t^2<2^{126}$, which holds for $B_t<2^{47}$.
\end{proof}
For the model in this study ($d=2{,}319$) the largest tested quantity is below $2^{119}$, which leaves seven bits of headroom. With every coordinate at the bound $|u_{i,k}|=2^{20}$, the secure tests and the plaintext rule returned the same bit in all eight tested cases ($\tau\in\{0.3,0,-0.3\}$, accepting and rejecting bounds).

\subsection*{Root-set reference and reference capture}

The reference $r_t$ is a local update computed on the root set; the servers know it, but the sites do not. The alternative, the previous aggregate, is public and shaped by whatever the servers accepted. An adaptive attacker can steer it so that attackers are accepted and honest sites rejected, a feedback loop we call reference capture.

\subsection*{Verifiable aggregation}

Each site posts vector Pedersen commitments\cite{pedersen1991} $C_i^j=\Com(s_i^j;\rho_i^j)$, $j\in\{1,2\}$, to the bulletin board and gives server $j$ its share $s_i^j$ and randomness $\rho_i^j$; server $j$ checks the opening privately. Each server publishes a Merkle root\cite{merkle1987} over its commitments, the accept bitmap, the aggregate share $A_j=\sum_{i:\,b_i=1}s_i^j$ and the summed randomness $R_j$. Anyone can check that the two bitmaps agree and that
\begin{equation}\label{eq:audit}
  \prod_{i:\,b_i=1} C_i^j = \Com(A_j;R_j), \qquad j\in\{1,2\},
\end{equation}
by the homomorphism of the commitment, and each site checks the Merkle path of its own commitment. A server that drops an accepted contribution, or publishes a bitmap that differs from the other server's, fails \eqref{eq:audit} or the bitmap check; unit tests cover both cases. The audit does not certify the accept decisions themselves. Shares carry no message-authentication codes, so a server that deviates inside the secure tests can change a decision consistently on both bitmaps; this behaviour is excluded by item (ii) of Definition~\ref{def:threat}.

\subsection*{Security and privacy analysis (informal)}

The following three propositions characterize the security guarantees of PVBRF-ID under Definition~\ref{def:threat}. Because secure comparison and per-coordinate range proofs are idealized functionalities, the arguments below are proof sketches rather than full reductions; they identify the cryptographic assumptions each claim rests on.

\begin{proposition}[Single-server privacy]\label{prop:privacy}
Under Definition~\ref{def:threat}(ii), an honest-but-curious server $S_j$ cannot distinguish two collections of site updates $\{u_i\}$ and $\{u_i'\}$ that produce the same accept bitmap $\{b_i\}$ and the same aggregate $\sum_{i:b_i=1} u_i$.
\end{proposition}
\begin{proof}[Sketch]
$S_j$ observes: (a) its own shares $\{s_i^j\}$, (b) the accept bits $\{b_i\}$ revealed by the ideal comparison functionality, (c) its aggregate share $A_j$, and (d) the public transcript. Share $s_i^1$ is drawn uniformly at random from $\mathbb{F}_Q^d$ independently of $u_i$; consequently, the marginal distribution of $S_1$'s share vector is $\mathcal{U}(\mathbb{F}_Q^d)$ regardless of $u_i$. A simulator that samples each $s_i^1$ uniformly and sets $A_1 = \sum_{i:b_i=1} s_i^1$ produces a view identically distributed to the real protocol. $S_2$'s analysis is symmetric. The only residual leakage to $S_j$ is the number of accepted sites $|\mathcal{A}_t|$ (from $|\{b_i:b_i=1\}|$) and the published aggregate, both intentionally revealed. The ideal comparison functionality leaks nothing beyond the accept bit by assumption.
\end{proof}

\begin{proposition}[Bounded Byzantine influence]\label{prop:robust}
Under Definition~\ref{def:threat}(i), let $f < n/2$ sites be Byzantine. The $\ell_2$ displacement of the aggregate update caused by all accepted Byzantine sites satisfies
\[
  \Bigl\|\sum_{i \in \mathcal{A}_t \cap \mathcal{F}} u_i \Bigr\| \;\le\; f \cdot B_t \;=\; f\kappa\|r_t\|,
\]
where $\mathcal{F}$ is the Byzantine set. Consequently, the per-round displacement of the averaged update $\Delta_t$ attributable to Byzantine sites is at most $f\kappa\|r_t\|/|\mathcal{A}_t|$ in any direction, where $|\mathcal{A}_t|$ is the number of accepted sites.
\end{proposition}
\begin{proof}[Sketch]
For any accepted site $i$, the norm test in \eqref{eq:accept} enforces $\|u_i\| \le B_t$. An accepted Byzantine site therefore contributes at most $B_t$ in $\ell_2$ norm. Summing over at most $f$ accepted Byzantine sites and dividing by the number of accepted sites $|\mathcal{A}_t|$, which is the denominator of $\Delta_t$ in Algorithm~\ref{alg:pvbrf}, gives the bound. The bound weakens when honest sites are rejected, which is why the honest rejection rate is reported throughout the Results section. The range proof (step 5 of Algorithm~\ref{alg:pvbrf}) prevents a Byzantine site from encoding an out-of-range value that passes the norm test through fixed-point wraparound; this is guaranteed formally by Lemma~\ref{lem:wrap}. The cosine test in \eqref{eq:accept}--\eqref{eq:tolerant} further restricts the direction of any accepted update, but the norm bound alone suffices for this proposition.
\end{proof}

\begin{proposition}[Audit soundness]\label{prop:audit}
Under Definition~\ref{def:threat}(iii), suppose a server publishes $(A_j, R_j, \{b_i\}, \mathrm{root}_j)$ such that equation~\eqref{eq:audit} holds for $j \in \{1,2\}$ and all Merkle paths verify. Then $A_j = \sum_{i:\,b_i=1} s_i^j \pmod{Q}$, except with probability negligible in the security parameter of the Pedersen commitment scheme.
\end{proposition}
\begin{proof}[Sketch]
Each commitment $C_i^j = \mathsf{Com}(s_i^j;\,\rho_i^j)$ was posted to the bulletin board before the round began and verified by the server at opening (step 4 of Algorithm~\ref{alg:pvbrf}). The homomorphic property of Pedersen commitments gives $\prod_{i:\,b_i=1} C_i^j = \mathsf{Com}(\sum_{i:\,b_i=1} s_i^j;\, \sum \rho_i^j)$. Equation~\eqref{eq:audit} asserts this equals $\mathsf{Com}(A_j;\,R_j)$. Because Pedersen commitments are computationally binding under the discrete-logarithm assumption, the only $(A_j, R_j)$ satisfying \eqref{eq:audit} for fixed $\{C_i^j\}$ and $\{b_i\}$ is the honest one. A server that drops an accepted share $s_i^j$ produces a product that is short by $C_i^j$, breaking \eqref{eq:audit}. Merkle path verification ensures the commitments were posted before the accept decisions, preventing a server from substituting commitments after seeing the bitmap.
\end{proof}

\subsection*{Optional mitigations}

We evaluate three optional mitigations. (i) Secure clipping rescales each accepted update to norm at most $C_t=\norm{r_t}$. (ii) Gaussian noise with standard deviation $zC_t/|\mathcal{A}_t|$ per coordinate is added inside the aggregate, where $\mathcal{A}_t$ is the accepted set and $z$ the noise multiplier. (iii) The aggregate clip bounds each coordinate $k$ of the revealed aggregate by $3|r_{t,k}|+0.5\,\mathrm{rms}(r_t)$, a public post-processing step. We report the privacy budget by R\'enyi-DP accounting\cite{mironov2017rdp} of the Gaussian noise over $T$ rounds, with sensitivity $C_t$, client-level add/remove adjacency and full participation:
\begin{equation}\label{eq:eps}
  \varepsilon = \min_{\alpha>1}\left( \frac{T\alpha}{2z^2} + \frac{\ln(1/\delta)}{\alpha-1} \right).
\end{equation}
The accept bits and $|\mathcal{A}_t|$ depend on the data and are outside the accounting in \eqref{eq:eps}.

\subsection*{Implemented and idealized components}

Secret sharing, the Beaver-based norm and cosine tests including the tolerant case, SHA-256 Merkle trees, Pedersen commitments and the homomorphic audit are implemented and covered by 12 unit tests, including tamper detection. Idealized are the secure comparison (it reveals only the accept bit, and its communication is a placeholder), the per-coordinate range proofs, triple generation by a trusted dealer, and secure clipping with noise addition inside the aggregate. The Pedersen group is a toy 512-bit group, so timings are upper bounds and the group is not secure. In the training runs, commitments and opening checks are switched off to save time; they do not change the accept bits or the aggregate, and the cost benchmark runs them in full.

\subsection*{Data and features}

We used the DNN file of Edge-IIoTset\cite{ferrag2022edge}: 2,219,201 flows in 15 classes (Normal and 14 attack types). Each flow was retained independently with probability $\min(1, c/N_k)$, where $N_k$ is the class count in the file and the cap $c$ is 100,000 for Normal and 10,000 for each attack class, so the caps are expected rather than exact counts; the two rare classes (MITM, 1,214 flows; Fingerprinting, 1,001 flows) were kept in full. The sample contains 222,947 flows (55\% attack; per-class counts in Supplementary Table~\ref{tab:s1}). Identifiers (addresses, timestamps and source ports) were dropped. Each text field yields a presence flag, its log length, its share of special characters, and signature flags for SQL injection, cross-site scripting, path traversal and file upload. Low-cardinality fields are one-hot encoded, and destination ports are encoded as top-15 one-hot indicators plus a high-port flag. The 80 features with the highest one-vs-rest correlation are kept; feature ranking uses all sampled rows, which introduces a small optimistic bias in feature selection that affects all conditions equally and does not alter the comparisons between aggregation rules. Per class, 20\% of flows are held out for testing, and a stratified root set of 1,500 flows is withheld from training.

\subsection*{Sites, model and training}

Ten sites receive class-skewed partitions drawn from a Dirichlet distribution with concentration 0.5 over classes. Across seeds, sites held 6,572 to 38,532 flows with attack shares from 14\% to 95\%. The model is a multilayer perceptron with layer sizes 80, 24 and 15 ($d=2{,}319$ parameters), trained with 30 local stochastic gradient descent (SGD) steps per round (batch 128, learning rate 0.1) for 30 rounds. Supplementary Table~\ref{tab:s1} lists all parameters.

\subsection*{Baselines}

We compare against centralized training on pooled data (reference), local-only training, FedAvg\cite{mcmahan2017fedavg}, coordinate-wise trimmed mean\cite{yin2018trimmed} (three values trimmed per side), multi-Krum\cite{blanchard2017krum} (assumed $f=3$), secure norm-only aggregation (PVBRF-ID without the cosine test, same $B_t$), PVBRF-ID with the previous-aggregate reference, and PVBRF-ID with the root reference at $\tau=0$ and $\tau=-0.3$. Trimmed mean and multi-Krum are configured for up to three attackers in every scenario, including clean runs.

\subsection*{Attacks}

Two of ten sites (20\%) attack, or three (30\%) where stated. Sign-flip submits $-4u$. Label-flip maps each label $y$ to $(y+1)\bmod15$. Targeted flip relabels all attack classes as Normal. Backdoor sets a trigger on the three lowest-ranked features of 50\% of attack samples and relabels them Normal; ASR is the share of triggered attack test flows classified Normal. Boosted backdoor scales the same update by three. The adaptive attacker knows $B_t$, $\tau$ and the reference. It submits an update of norm $0.95B_t$ along the loss-ascent direction $-u_{\text{honest}}$, at cosine $\tau+0.02$ to the reference: 0.02 for $\tau=0$ and $-0.28$ for $\tau=-0.3$, where it exploits the tolerance. Supplementary Note~\ref{note:s1} documents a correction to this attacker.

\subsection*{Privacy attacker}

In no-attack runs, during the first 10 rounds, the attacker sees either the per-site updates (plaintext baselines) or one server's share vectors plus the accept bits (PVBRF-ID); the colluding view combines both shares. Class-mix recovery is the per-site Spearman correlation between the centred output-bias update and the centred class histogram. A top-3 hit occurs when the site's most over-represented class is among the attacker's top three (chance 0.20). Similarity is the correlation between output-bias update similarity and class-mix similarity over site pairs. A neighbour hit occurs when the attacker's most similar site is among the three truly most similar (chance 0.33). The accept-bit correlation, over all 30 rounds, relates a site's acceptance count to its atypicality, the total-variation distance between its class histogram and the root distribution. This attacker is generic.

\subsection*{Metrics and statistics}

Macro-F1 on the held-out test set is the primary metric; backdoor ASR and $\varepsilon$ are secondary. We report means with t-based 95\% confidence intervals over 10 seeds, and comparisons are paired by seed (same partition and test split).

\textit{Reproducibility.} Experiments were run with Python 3.13.9, NumPy 2.3.5 and pandas 2.3.3 on Windows 11 (64-bit). Total wall-clock time was 1.1--2.9 hours per configuration depending on whether commitments were enabled. All 870 main runs and 160 supplementary runs are deterministic given the seeds listed in Supplementary Table~\ref{tab:s1}; the four plaintext baselines, which are unaffected by the tolerance-exploiting attacker, reproduced identically across the two runs on every seed. Code, raw per-run results and the figure script are available at \url{https://github.com/nawaralseelawi/pvbrf-id}.

\section*{Results}

\subsection*{Federation recovers rare attack classes}

FedAvg reached macro-F1 $0.706\pm0.008$, against $0.536\pm0.019$ for local-only training (paired difference $+0.170\pm0.022$, all 10 seeds), and came within $0.010\pm0.012$ of centralized training (0.716; Table~\ref{tab:f1}). The gain concentrated in classes that a site rarely sees locally: recall rose for MITM from 0.73 to 1.00, for DDoS-ICMP from 0.77 to 0.99 and for Password from 0.38 to 0.61 (per-class recall across all 15 classes in Supplementary Table~\ref{tab:s3}; local-only figures are from site 0's model as a representative site). For an operator, these classes are the attacks a single wellhead or gathering station has never observed, and federation is what makes them detectable at that site. Macro-F1 per round for clean and adaptive 30\% runs is shown in Fig.~\ref{fig:convergence}.

\begin{figure}[t]
  \centering
  \includegraphics[width=\linewidth]{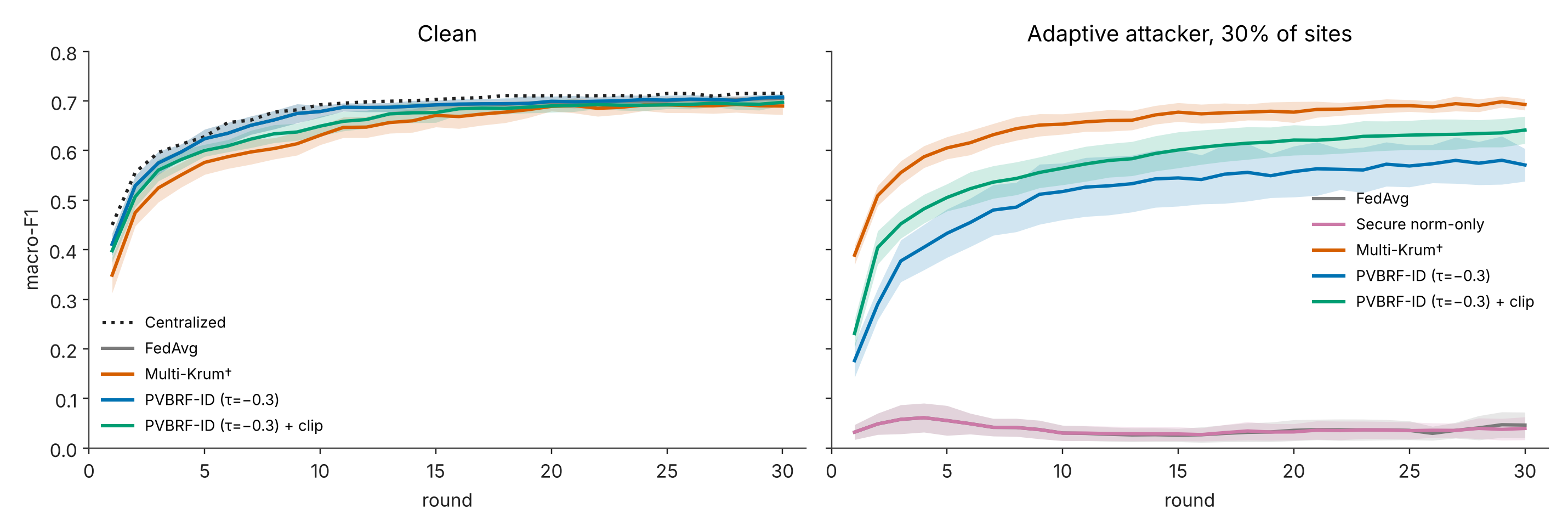}
  \caption{\textbf{Per-round macro-F1 during training.} Left: clean training; right: adaptive attacker controlling 30\% of sites. Lines show means over 10 seeds; shaded bands show t-based 95\% confidence intervals. PVBRF-ID ($\tau=-0.3$) with secure clipping (orange) converges at a similar pace to FedAvg (blue) on clean data and holds substantially higher accuracy than FedAvg under the adaptive attacker.}
  \label{fig:convergence}
\end{figure}

\subsection*{Non-adaptive attacks}

Sign-flip destroyed FedAvg (0.091), while PVBRF-ID with the root reference kept $0.696\pm0.013$ ($+0.605\pm0.030$ over FedAvg; Table~\ref{tab:f1}, Fig.~\ref{fig:f1}). PVBRF-ID was on par with multi-Krum ($+0.010\pm0.013$) and above trimmed mean ($+0.046\pm0.023$). On clean data PVBRF-ID was indistinguishable from FedAvg ($+0.002\pm0.005$) and above multi-Krum and trimmed mean ($+0.018\pm0.019$ and $+0.018\pm0.014$). Label-flip and targeted flip were weak at 20\% adversaries: FedAvg lost 0.033 and 0.017. PVBRF-ID rejected 90\% and 87\% of their updates, but there was little damage to prevent.

\begin{table}[t]
\centering
\caption{\textbf{Macro-F1 (mean $\pm$ t-based 95\% confidence interval over 10 seeds).} PVBRF-ID: proposed protocol; $\dagger$, the rule reads plaintext updates; --, not run. Adaptive columns use the tolerance-exploiting attacker for the $\tau=-0.3$ rows; rules without a reference are unaffected by the tolerance. Shaded rows: recommended designs.}
\label{tab:f1}
\footnotesize
\setlength{\tabcolsep}{3pt}
\resizebox{\linewidth}{!}{%
\begin{tabular}{@{}lcccccccc@{}}
\toprule
Method & Clean & Sign-flip & Label-flip & Targeted flip & Backdoor & Backdoor $\times3$ & Adaptive 20\% & Adaptive 30\% \\
\midrule
Centralized (reference) & $0.716\pm0.005$ & -- & -- & -- & -- & -- & -- & -- \\
Local-only & $0.536\pm0.019$ & -- & -- & -- & -- & -- & -- & -- \\
FedAvg & $0.706\pm0.008$ & $0.091\pm0.032$ & $0.673\pm0.015$ & $0.689\pm0.014$ & $0.693\pm0.011$ & $0.704\pm0.007$ & $0.393\pm0.062$ & $0.046\pm0.026$ \\
Trimmed mean$^\dagger$ & $0.690\pm0.012$ & $0.650\pm0.023$ & $0.677\pm0.013$ & $0.675\pm0.014$ & $0.685\pm0.013$ & $0.688\pm0.012$ & $0.649\pm0.023$ & $0.597\pm0.030$ \\
Multi-Krum$^\dagger$ & $0.690\pm0.017$ & $0.686\pm0.022$ & $0.686\pm0.022$ & $0.687\pm0.022$ & $0.683\pm0.017$ & $0.686\pm0.016$ & $0.683\pm0.021$ & $0.693\pm0.011$ \\
Secure norm-only & $0.706\pm0.008$ & $0.691\pm0.014$ & $0.677\pm0.017$ & $0.689\pm0.014$ & $0.693\pm0.011$ & $0.688\pm0.015$ & $0.398\pm0.063$ & $0.039\pm0.023$ \\
PVBRF-ID, previous-aggregate ref.\ ($\tau=0$) & $0.670\pm0.031$ & -- & -- & -- & -- & -- & $0.012\pm0.009$ & $0.017\pm0.020$ \\
PVBRF-ID, root ref.\ ($\tau=0$) & $0.701\pm0.014$ & $0.684\pm0.021$ & $0.687\pm0.014$ & $0.691\pm0.014$ & $0.696\pm0.008$ & $0.682\pm0.018$ & $0.677\pm0.019$ & $0.640\pm0.028$ \\
\rowcolor{recrow} PVBRF-ID, root ref.\ ($\tau=-0.3$) & $0.708\pm0.006$ & $0.696\pm0.013$ & $0.695\pm0.010$ & $0.696\pm0.012$ & $0.693\pm0.011$ & $0.689\pm0.012$ & $0.645\pm0.022$ & $0.571\pm0.033$ \\
\rowcolor{recrow} PVBRF-ID ($\tau=-0.3$) + clip & $0.697\pm0.006$ & $0.688\pm0.011$ & -- & -- & $0.688\pm0.012$ & $0.680\pm0.013$ & $0.665\pm0.019$ & $0.641\pm0.027$ \\
PVBRF-ID ($\tau=-0.3$) + clip + noise & $0.696\pm0.012$ & $0.684\pm0.013$ & -- & -- & $0.685\pm0.009$ & $0.677\pm0.014$ & $0.665\pm0.020$ & $0.643\pm0.027$ \\
\bottomrule
\end{tabular}}
\end{table}

\begin{figure}[t]
  \centering
  \includegraphics[width=\linewidth]{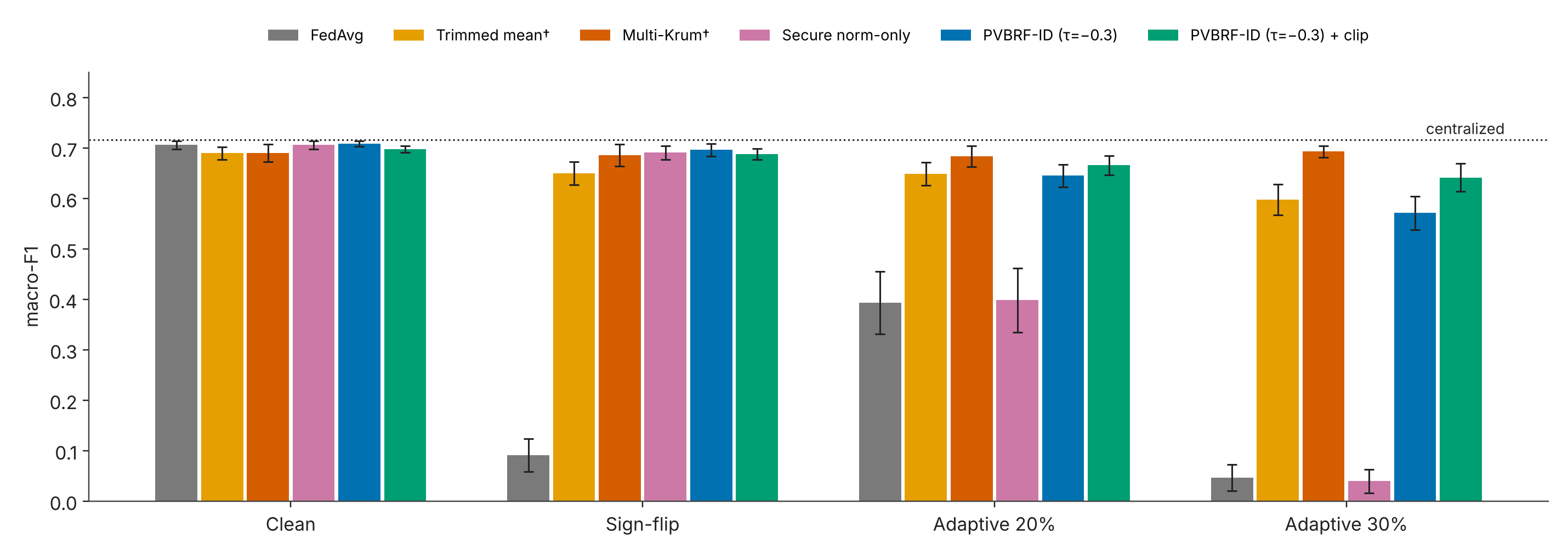}
  \caption{\textbf{Macro-F1 under clean training, sign-flip and adaptive attacks.} Bars show means over 10 seeds; error bars show t-based 95\% confidence intervals. The dotted line marks centralized training. $\dagger$, the rule reads plaintext updates. Adaptive 20\% and 30\% denote two and three attacking sites out of ten.}
  \label{fig:f1}
\end{figure}

\subsection*{Adaptive attackers and reference capture}

With the previous aggregate as the reference, PVBRF-ID collapsed under the adaptive attacker: macro-F1 fell to $0.012\pm0.009$ at 20\% adversaries and $0.017\pm0.020$ at 30\%, and PVBRF-ID rejected 85\% and 97\% of honest updates. This design was already weaker on clean data ($0.670\pm0.031$, 21\% honest rejections). With the root set as the reference, the collapse did not occur.

Under the tolerance-exploiting attacker, PVBRF-ID with $\tau=-0.3$ reached $0.645\pm0.022$ at 20\% and $0.571\pm0.033$ at 30\% adversaries. FedAvg reached 0.393 and 0.046, and norm-only aggregation 0.398 and 0.039 (PVBRF-ID advantage over norm-only: $+0.247\pm0.074$ and $+0.531\pm0.053$). The cosine test, not the norm bound, supplies the robustness. PVBRF-ID was comparable to trimmed mean ($-0.004\pm0.024$ and $-0.026\pm0.034$) but below multi-Krum ($-0.038\pm0.031$ and $-0.122\pm0.035$). The attackers were never rejected, because they sit inside the acceptance cone by construction: robustness comes from bounding their influence, not from detecting them.

The stricter threshold $\tau=0$ rejects more honest updates as training converges. At rounds 20--30 with $\tau=0$, PVBRF-ID rejected 32--33\% of honest updates as their updates gradually aligned with the root reference; $\tau=-0.3$ held rejections to 2--5\% across the same rounds. This late-training rejection is the main cost of tightening the threshold.

The tolerance has a measurable cost. Exploiting $\tau=-0.3$ cost PVBRF-ID $0.034\pm0.018$ and $0.075\pm0.034$ macro-F1 relative to an attacker that does not exploit it (10/10 seeds). Setting $\tau=0$ is more robust under the same attack ($+0.032\pm0.018$ and $+0.069\pm0.033$) but rejected 16.5\% of honest updates in clean runs, against 0.9\% at $\tau=-0.3$. Secure clipping recovered most of the loss at 30\% adversaries: $0.641\pm0.027$, which is $+0.071\pm0.031$ over unclipped PVBRF-ID (10/10 seeds), $-0.052\pm0.032$ against multi-Krum and $+0.044\pm0.031$ against trimmed mean, at a clean-data cost of 0.011.

\subsection*{Backdoors}

The filter did not stop backdoors. The attack success rate (ASR) was 1.000 for PVBRF-ID, FedAvg and norm-only aggregation, because backdoor updates are small and well aligned with the reference (Table~\ref{tab:backdoor}). Trimmed mean reached $0.054\pm0.051$, and multi-Krum was bimodal ($0.401\pm0.369$; ASR 1.0 in five seeds and about 0 in the other five). Boosting the backdoor update by a factor of three passed the loose bound $\kappa=3.5$ (PVBRF-ID $0.872\pm0.076$; norm-only $0.901\pm0.071$), and secure clipping alone removed it ($0.048\pm0.048$).

Against plain backdoors, clipping or noise alone had no effect (ASR 0.99). Clipping, noise and the aggregate clip together lowered ASR to 0.411, 0.150 and 0.091 at noise multipliers $z=0.5$, 1 and 2, at clean macro-F1 of 0.689, 0.680 and 0.615 (baseline 0.708; Fig.~\ref{fig:backdoor}). At $z=1$ the paired change was $-0.849\pm0.069$ in ASR and $-0.029\pm0.015$ in macro-F1, and trimmed mean remained better by $0.096\pm0.069$. The corresponding privacy budgets ($\varepsilon = 112.6$, 41.3 and 16.9 at $\delta=10^{-5}$ over 30 rounds) are too large to be a meaningful guarantee. We therefore present the noise as a heuristic backdoor defence, not as differential privacy. In operational terms, at $z=1$ about one triggered attack flow in seven is still classified as Normal.

\begin{table}[t]
\centering
\caption{\textbf{Backdoor mitigation.} Clean F1 and Adapt.\ 30\% F1 are macro-F1; ASR is the share of triggered attack test flows classified as Normal (lower is better). Adapt.\ 30\% F1 uses the tolerance-exploiting attacker. $\varepsilon$: R\'enyi-DP accounting of the Gaussian noise only, $\delta=10^{-5}$, 30 rounds; accept bits are not covered. $\dagger$, the rule reads plaintext updates; --, not run or not applicable.}
\label{tab:backdoor}
\footnotesize
\begin{tabular}{@{}lccccc@{}}
\toprule
Method & Clean F1 & Backdoor ASR & Backdoor $\times3$ ASR & Adapt.\ 30\% F1 & $\varepsilon$ \\
\midrule
Trimmed mean$^\dagger$ & 0.690 & $0.054\pm0.051$ & $0.056\pm0.054$ & 0.597 & -- \\
Multi-Krum$^\dagger$ & 0.690 & $0.401\pm0.369$ & $0.004\pm0.008$ & 0.693 & -- \\
PVBRF-ID ($\tau=-0.3$) & 0.708 & $1.000\pm0.001$ & $0.872\pm0.076$ & 0.571 & -- \\
\quad + clip & 0.697 & $0.993\pm0.009$ & $0.048\pm0.048$ & 0.641 & -- \\
\quad + clip + noise, $z=0.5$ & 0.696 & $0.992\pm0.007$ & $0.162\pm0.204$ & 0.643 & 113 \\
\quad + clip + noise + agg.\ clip, $z=0.5$ & 0.689 & $0.411\pm0.193$ & $0.009\pm0.017$ & -- & 113 \\
\quad + clip + noise + agg.\ clip, $z=1$ & 0.680 & $0.150\pm0.069$ & $0.002\pm0.002$ & 0.592 & 41 \\
\quad + clip + noise + agg.\ clip, $z=2$ & 0.615 & $0.091\pm0.034$ & $0.025\pm0.022$ & -- & 17 \\
\bottomrule
\end{tabular}
\end{table}

\begin{figure}[t]
  \centering
  \includegraphics[width=0.62\linewidth]{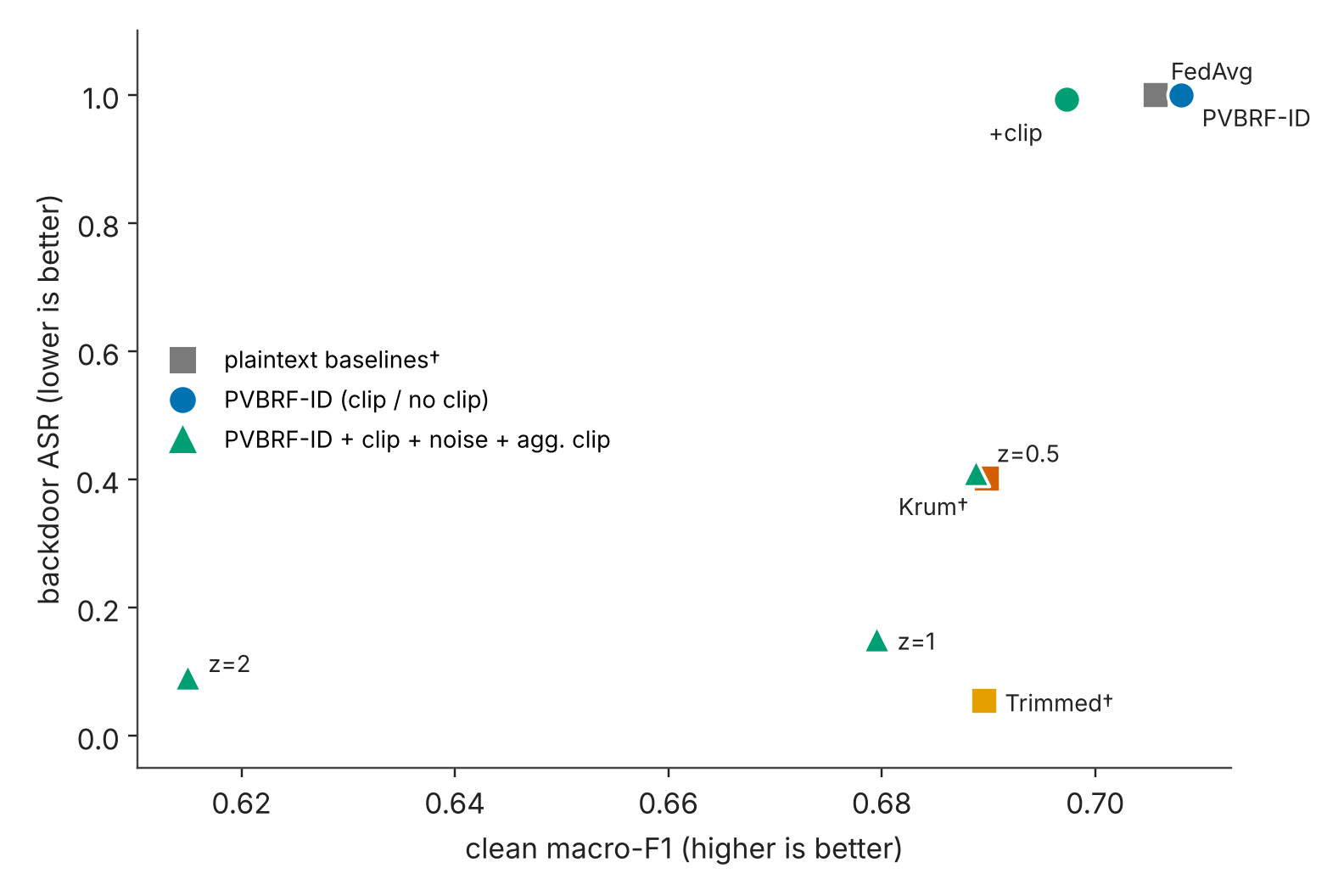}
  \caption{\textbf{Backdoor attack success rate against clean macro-F1.} Each point is the mean over 10 seeds. Squares: plaintext baselines ($\dagger$); circles: PVBRF-ID with and without secure clipping; triangles: PVBRF-ID with clipping, Gaussian noise (multiplier $z$) and the aggregate clip.}
  \label{fig:backdoor}
\end{figure}

\subsection*{Privacy against update inference}

From plaintext updates, the aggregator recovered each site's class mix with Spearman correlation $0.81\pm0.02$ and placed the site's dominant class in its top three 70\% of the time (chance 20\%; Table~\ref{tab:privacy}, Fig.~\ref{fig:privacy}). Neighbour identification succeeded 60\% of the time (chance 33\%). For an operator, this means a plaintext aggregator learns which attack types each field site experiences, which is itself sensitive security information. Given one server's shares, the same attacker scored $0.00\pm0.06$ and 0.20. This result is expected by construction, because a single additive share is uniformly distributed, and it serves as an end-to-end check of the implementation. With both servers colluding, the attacker regained plaintext-level success, so non-collusion is an assumption of PVBRF-ID, not a result.

The channels that remain open to one server are the accept bits and the revealed aggregate. The accept bit leaked little at $\kappa=3.5$: its correlation with site atypicality was $-0.04\pm0.25$, honest rejection was below 1\%, and three of ten seeds had no rejections. Whole-update similarity was a weak attack (correlation 0.24 for FedAvg and about 0.03 for multi-Krum), so we report the output-bias variant.

\begin{table}[t]
\centering
\caption{\textbf{Update-inference attacker success} (mean $\pm$ t-based 95\% confidence interval over 10 seeds, no-attack runs). Similarity $\rho$: output-bias update similarity against class-mix similarity over site pairs. Accept-bit $\rho$: acceptance count against atypicality over 30 rounds (defined in 7 of 10 seeds). The one-server row is uniform by construction and checks the implementation.}
\label{tab:privacy}
\footnotesize
\begin{tabular}{@{}lccccc@{}}
\toprule
View & Class-mix $\rho$ & Top-3 hit & Similarity $\rho$ & Neighbour hit & Accept-bit $\rho$ \\
\midrule
FedAvg (plaintext) & $0.81\pm0.02$ & $0.70\pm0.10$ & $0.27\pm0.05$ & $0.60\pm0.10$ & -- \\
Trimmed mean (plaintext) & $0.81\pm0.02$ & $0.69\pm0.10$ & $0.27\pm0.06$ & $0.57\pm0.08$ & -- \\
Multi-Krum (plaintext) & $0.81\pm0.02$ & $0.66\pm0.06$ & $0.31\pm0.07$ & $0.64\pm0.11$ & -- \\
PVBRF-ID: one server's share & $0.00\pm0.06$ & $0.20\pm0.08$ & $0.01\pm0.06$ & $0.29\pm0.10$ & $-0.04\pm0.25$ \\
PVBRF-ID: both servers collude & $0.81\pm0.02$ & $0.70\pm0.10$ & $0.27\pm0.05$ & $0.60\pm0.10$ & -- \\
Chance level & 0 & 0.20 & 0 & 0.33 & 0 \\
\bottomrule
\end{tabular}
\end{table}

\begin{figure}[t]
  \centering
  \includegraphics[width=0.62\linewidth]{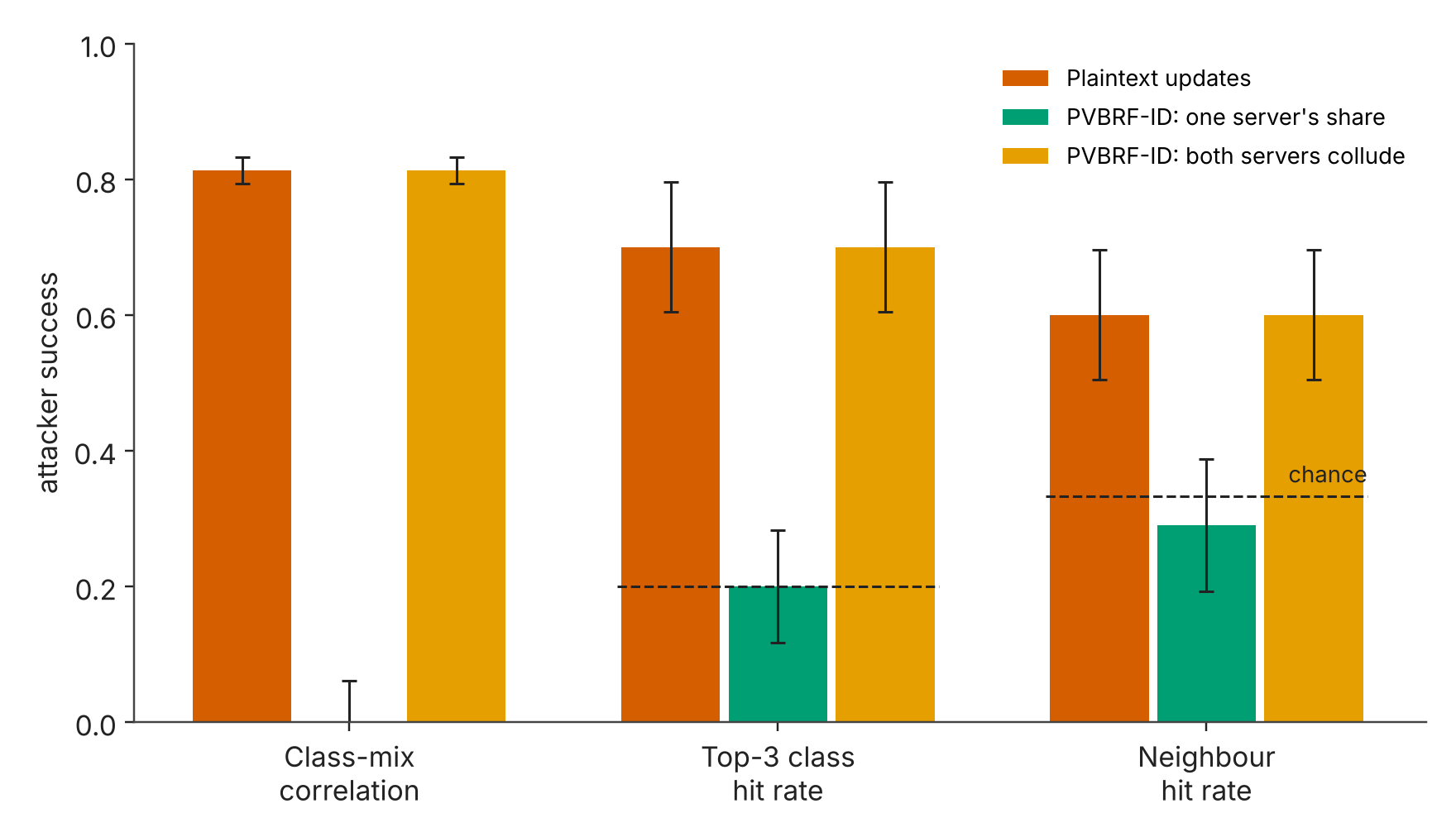}
  \caption{\textbf{Update-inference attacker success.} Bars show means over 10 seeds; error bars show t-based 95\% confidence intervals; dashed lines mark chance level. Plaintext updates, one server's share and the colluding view of both servers are compared on class-mix correlation, top-3 class hit rate and neighbour hit rate.}
  \label{fig:privacy}
\end{figure}

\subsection*{Fairness and the norm bound}

A tight norm bound rejects honest sites with unusual class mixes and leaks their atypicality through the accept bit (Table~\ref{tab:kappa}). At $\kappa=1.5$, PVBRF-ID rejected 43.5\% of honest updates. The most atypical third of sites was accepted 33.7\% of the time against 82.4\% for the most typical third (Spearman $\rho=-0.52$). At $\kappa=3.5$ the two acceptance rates were 99.6\% and 99.2\% ($\rho=-0.01$), and $\kappa=5$ was statistically identical to $\kappa=3.5$. Multi-Krum accepted atypical sites 40.7\% of the time against 93.5\% for typical ones ($\rho=-0.62$). It systematically excludes the sites that carry rare attack classes, which are the sites whose data the federation most needs.

\begin{table}[t]
\centering
\caption{\textbf{Norm bound $\kappa$, honest rejection and fairness} (clean runs unless stated). Accept typ./atyp.: acceptance rate of the third of sites closest to / farthest from the root class distribution. Rank $\rho$: Spearman rank correlation of acceptance count against atypicality. $\dagger$, the rule reads plaintext updates.}
\label{tab:kappa}
\footnotesize
\begin{tabular}{@{}lcccccc@{}}
\toprule
Method & Clean F1 & Sign-flip F1 & Honest rej.\ (\%) & Accept typ.\ (\%) & Accept atyp.\ (\%) & Rank $\rho$ \\
\midrule
PVBRF-ID, $\kappa=1.5$ & 0.688 & 0.683 & 43.5 & 82.4 & 33.7 & $-0.52$ \\
PVBRF-ID, $\kappa=2.5$ & 0.703 & 0.693 & 9.2 & 95.8 & 86.0 & $-0.30$ \\
PVBRF-ID, $\kappa=3.5$ & 0.708 & 0.696 & 0.9 & 99.2 & 99.6 & $-0.01$ \\
PVBRF-ID, $\kappa=5.0$ & 0.708 & 0.696 & 0.8 & 99.2 & 99.7 & $+0.02$ \\
Multi-Krum$^\dagger$ & 0.690 & 0.686 & 30.0 & 93.5 & 40.7 & $-0.62$ \\
\bottomrule
\end{tabular}
\end{table}

\subsection*{Communication cost}

At model dimension $d=2{,}319$ and ten sites, online communication was 1.72~MiB per round, 19.4 times the upload of plain FedAvg (Supplementary Table~\ref{tab:s2}). Excluding the placeholder cost of the ideal comparison, the ratio was 16.8; adding an estimate for ideal secure clipping raised it to about 27.4. Offline Beaver triples added 0.71~MiB per round. Upload alone was eight times the plain upload, because each coordinate is sent as two 16-byte shares. The audit and all inclusion proofs verified. Commitment verification took 8.1~s and the audit 0.8~s per round in a toy 512-bit group implemented in Python; these timings are upper bounds.

\section*{Discussion}

Where updates may not be exposed to the aggregator, PVBRF-ID keeps accuracy at the FedAvg level on clean data. With secure clipping, it stays within 0.02 to 0.05 macro-F1 of plaintext multi-Krum at 20\% to 30\% adversaries. It exposes one accept bit per site and a verifiable transcript, and it does not penalize atypical sites. The accept bit is a deliberate trade: FLOD\cite{dong2021flod} reveals only the aggregate, whereas PVBRF-ID reveals each decision so that sites and auditors can check inclusion. The price is about 19 to 27 times the communication of FedAvg, a second administrative domain and a clean root set. In an oilfield setting, the second domain could be a contractor or regional security operations centre, and the root set could come from a testbed capture. Where plaintext access to updates is acceptable, multi-Krum is the more robust choice at 30\% adversaries, at the cost of excluding atypical sites.

\textit{Deployment considerations.} PVBRF-ID targets settings where two independently operated servers are available---a corporate security operations centre and a regional or contractor domain is a natural split in the oilfield context---and where a small testbed capture can supply the root set. The protocol adds roughly 19 times the FedAvg communication overhead, which is feasible on corporate WAN links but would be expensive on low-bandwidth satellite or radio channels to remote wellheads. In such settings, the root-guided cosine test could run on plaintext updates at the cost of privacy, or less frequent rounds (e.g., daily rather than hourly) could amortize the bandwidth cost. Table~\ref{tab:goals} summarises the five design goals and the evidence from the evaluation.

\begin{table}[t]
\centering
\caption{\textbf{Design goals and evaluation evidence.}}
\label{tab:goals}
\footnotesize
\begin{tabular}{@{}lp{5cm}l@{}}
\toprule
Goal & Description & Evidence and status \\
\midrule
G1 & Clean accuracy close to FedAvg & PVBRF-ID 0.708 vs.\ FedAvg 0.706 (\checkmark) \\
G2 & Robust to five attacks and adaptive attacker & Table~\ref{tab:f1}; multi-Krum better at 30\% (\checkmark/\textasciitilde) \\
G3 & One server learns $\le$ one accept bit per update & Table~\ref{tab:privacy}; $\rho=0.00\pm0.06$ (\checkmark) \\
G4 & Inclusion of accepted updates verifiable & Audit and Merkle inclusion pass; unit tests cover tamper (\checkmark) \\
G5 & Honest atypical sites not systematically excluded & Table~\ref{tab:kappa}, $\kappa=3.5$: $\rho=-0.01$ (\checkmark) \\
\bottomrule
\end{tabular}
\end{table}

The comparison with multi-Krum has one asymmetry. Our adaptive attacker targets PVBRF-ID's norm bound and reference, while multi-Krum and trimmed mean face only a scaled loss-ascent update. Attacks crafted against distance-based rules\cite{fang2020local} would test those baselines on equal terms, and we plan to add them. Robustness in PVBRF-ID also comes from bounded influence, not detection: attackers inside the acceptance cone are never rejected.

The evaluation uses one public dataset with simulated sites and a multilayer perceptron whose macro-F1 (0.71) equals the centralized ceiling of our training budget. Hard classes remain hard: Fingerprinting recall is about 0.10, Uploading 0.43 and SQL injection 0.49. A second IIoT dataset and a pilot on operator traffic would test whether the robustness ordering transfers. The norm bound $\kappa$ and threshold $\tau$ were chosen on this data; calibrating them on seeds held out from reporting removes this source of optimism.

Four cryptographic components are idealized: secure comparison, per-coordinate range proofs, triple generation and the clipping and noise steps inside the aggregate. Their costs are estimates, and the Pedersen group is a toy. Instantiating the comparison with an existing secure-comparison protocol\cite{damgard2009comparison} and generating triples without a dealer would turn the cost estimates into measurements. The audit binds published aggregate shares to committed inputs but does not certify the accept decisions themselves. A server that deviates inside the secure comparison step can change a decision while keeping both bitmaps consistent; authenticated shares or a ZK proof of the accept bit would close this gap\cite{baum2016abort}.

Privacy rests on non-colluding servers and does not cover leakage from the released global model. The inference attacker is generic; formalizing it and adding an attacker that combines accept bits with the aggregates across rounds would bound what one server learns. Backdoors are not stopped, and the noise defence is heuristic at $\varepsilon\geq16.9$.

Federated intrusion detection across oilfield sites must keep site updates private and still reject poisoned ones. We showed that secret-shared aggregation can keep a root-guided direction test: PVBRF-ID matched FedAvg on clean data (macro-F1 0.708), held 0.641 against a 30\% adaptive attacker with secure clipping where FedAvg fell to 0.046, and did not exclude atypical sites. Next steps are a real comparison protocol with zero-knowledge range proofs, stronger backdoor defences, a pilot on operator data, and a Zero Trust policy plane that consumes accept bits and norm statistics.

\clearpage
\appendix
\section*{Supplementary Information}
\renewcommand{\thetable}{S\arabic{table}}
\setcounter{table}{0}

\begin{table}[h]
\centering
\caption{\textbf{Experimental parameters} (notebook configuration, identical across all runs).}
\label{tab:s1}
\footnotesize
\begin{tabular}{@{}lp{11.5cm}@{}}
\toprule
Parameter & Value \\
\midrule
Dataset file & Edge-IIoTset, \texttt{DNN-EdgeIIoT-dataset.csv} (2,219,201 flows, 15 classes) \\
Sampling & per-flow retention probability $\min(1,c/N_k)$; $c=100{,}000$ (Normal), $10{,}000$ (each attack class); sampler seed 12345 \\
Sampled flows & 222,947 (Normal 100,172; attack classes 9,947--10,205 each; MITM 1,214; Fingerprinting 1,001) \\
Features & 80 highest one-vs-rest correlation (payload signature, length, special-character, one-hot and port features) \\
Split per class & 20\% test; stratified root set of 1,500 flows withheld from training; split seed $1000+s$ \\
Sites & 10, Dirichlet($\alpha=0.5$) class partition; minimum 300 flows and 5--95\% attack share per site \\
Site sizes (all seeds) & 6,572--38,532 flows; attack share 14--95\% \\
Model & MLP 80--24--15 (ReLU, softmax), $d=2{,}319$ parameters; He initialization \\
Local training & 30 SGD steps per round, batch 128, learning rate 0.1 \\
Rounds / seeds & 30 rounds; seeds $s\in\{0,\dots,9\}$ (training rng $10{,}000+s$, crypto rng $20{,}000+s$) \\
Norm bound & $\kappa=3.5$ (calibrated on 6 clean FedAvg rounds of seed 0: ratio median 1.22, p90 1.60, p99 1.93, max 2.07); sweep $\{1.5,2.5,5\}$ \\
Cosine threshold & $\tau=-0.3$ (recommended) and $\tau=0$; previous-aggregate reference at $\tau=0$ \\
Baselines & trimmed mean (3 per side); multi-Krum ($f=3$, 7 of 10 selected) \\
Attacks & 2 of 10 sites (3 for 30\%); sign-flip $-4u$; label-flip $y\!\to\!(y+1)\bmod 15$; targeted flip; backdoor (50\% of attack samples, trigger value 4.0 on 3 lowest-ranked features); boost $\times3$; adaptive at norm $0.95B_t$ \\
Mitigation & secure clip $C_t=\norm{r_t}$; Gaussian noise $z\in\{0.5,1,2\}$; aggregate clip $3|r_{t,k}|+0.5\,\mathrm{rms}(r_t)$; $\delta=10^{-5}$ \\
Privacy attacker & observes rounds 1--10 of no-attack runs; accept-bit correlation over all 30 rounds \\
Fixed-point / field & 16 fractional bits; $Q=2^{127}-1$; range bound $2^{20}$; toy 512-bit Pedersen group \\
Environment & Python 3.13.9, NumPy 2.3.5, pandas 2.3.3, Windows 11, 8 CPU cores \\
\bottomrule
\end{tabular}
\end{table}

\begin{table}[h]
\centering
\caption{\textbf{Communication per round} at $d=2{,}319$ and $n=10$ sites, measured by the protocol ledger. Plain FedAvg upload is $n\,d\times4$ bytes $=92{,}760$ bytes (float32).}
\label{tab:s2}
\footnotesize
\begin{tabular}{@{}p{8.6cm}rp{4.2cm}@{}}
\toprule
Phase & Bytes & Note \\
\midrule
Upload (two shares, two randomness values, two commitments per site) & 743,040 & $8.01\times$ plain upload \\
Secure tests: masked-vector openings for the norm test & 742,080 & one $d$-vector per server per site \\
Secure tests: scalar openings for the cosine test & 640 & \\
Secure tests: ideal comparison (placeholder, $2\times127$ elements, 8 rounds) & 243,840 & not measured, analytic estimate \\
Aggregate shares and summed randomness & 74,240 & \\
Publish (Merkle roots and bitmaps) & 68 & \\
\midrule
Online total & 1,803,908 & 1.72 MiB; $19.45\times$ plain upload \\
Online total excluding the comparison placeholder & 1,560,068 & $16.82\times$ \\
Online total plus an estimate for ideal secure clipping & 2,545,988 & $27.45\times$ \\
Offline (Beaver triples, trusted dealer) & 743,040 & 0.71 MiB \\
\midrule
Commitment verification (10 sites, two servers) & 7.7~s & toy 512-bit group, pure Python \\
Audit of both aggregate shares & 0.76~s & \\
\bottomrule
\end{tabular}
\end{table}

\begin{table}[h]
\centering
\caption{\textbf{Per-class recall on the held-out test set} (mean over 10 seeds, clean training). Local-only values are from site 0's model.}
\label{tab:s3}
\footnotesize
\begin{tabular}{@{}lcccc@{}}
\toprule
Class & Local-only & FedAvg & PVBRF-ID ($\tau=-0.3$) & Centralized \\
\midrule
Normal & 1.00 & 1.00 & 1.00 & 1.00 \\
Backdoor & 0.61 & 0.77 & 0.74 & 0.85 \\
DDoS-HTTP & 0.54 & 0.72 & 0.69 & 0.68 \\
DDoS-ICMP & 0.77 & 0.99 & 0.99 & 0.99 \\
DDoS-TCP & 0.69 & 0.84 & 0.80 & 0.71 \\
DDoS-UDP & 0.82 & 1.00 & 1.00 & 1.00 \\
Fingerprinting & 0.14 & 0.11 & 0.10 & 0.10 \\
MITM & 0.73 & 1.00 & 1.00 & 1.00 \\
Password & 0.38 & 0.61 & 0.62 & 0.58 \\
Port scanning & 0.64 & 0.69 & 0.73 & 0.82 \\
Ransomware & 0.53 & 0.60 & 0.63 & 0.52 \\
SQL injection & 0.60 & 0.49 & 0.49 & 0.54 \\
Uploading & 0.39 & 0.43 & 0.43 & 0.42 \\
Vulnerability scanner & 0.81 & 0.80 & 0.79 & 0.79 \\
XSS & 0.49 & 0.69 & 0.74 & 0.76 \\
\bottomrule
\end{tabular}
\end{table}

\paragraph{Supplementary Note S1. Correction to the adaptive attacker.}\label{note:s1}
In the main run (870 runs) the adaptive attacker placed its update at cosine $\max(\tau,0)+0.02$ to the reference, so for the tolerant threshold $\tau=-0.3$ it never used the tolerance and sat at cosine $0.02$. A worst-case attacker submits at cosine $\tau+0.02$, that is $-0.28$ for $\tau=-0.3$, partly against the reference while still inside the acceptance region of \eqref{eq:tolerant}. The supplementary run (160 runs) repeats only the adaptive scenarios with this attacker, on the same ten seeds, for the $\tau=-0.3$ designs and the four baselines. Because the baselines have no reference, their results must reproduce the main run exactly; they did, on every seed. Table~\ref{tab:f1} and Table~\ref{tab:backdoor} report the tolerance-exploiting attacker for all $\tau=-0.3$ rows. Relative to the non-exploiting attacker, exploiting the tolerance cost PVBRF-ID $0.034\pm0.018$ macro-F1 at 20\% and $0.075\pm0.034$ at 30\% adversaries (10 of 10 seeds in both cases).

\bibliographystyle{unsrtnat}
\bibliography{ref}

\end{document}